\documentclass[10pt,conference]{IEEEtran}

\usepackage[utf8]{inputenc}
\usepackage[table,dvipsnames]{xcolor}
\usepackage{graphicx}
\usepackage{amsmath,amssymb}
\usepackage{xfrac}
\usepackage{relsize}
\usepackage{optidef}
\usepackage{pgfplots}
\usepackage{tikzscale}
\usepackage{amsthm}
\usepackage{algorithm}
\usepackage{clipboard}
\usepackage{algorithmic}
\usepackage{nicematrix}
\usepackage{datetime}
\usepackage{microtype}
\usepackage{subcaption}
\usepackage[font=small,labelfont=bf]{caption}
\usepackage{pdfpages}

\newcommand\scalemath[2]{\scalebox{#1}{\mbox{\ensuremath{\displaystyle #2}}}}
\newcommand{\bo}[1]{\boldsymbol{#1}}
\newcommand{\ur}[1]{\mathrm{#1}}
\renewcommand{\H}{\ur{H}}
\newcommand{\T}{\ur{T}}

\newcommand{\btr}[1]{\ur{tr}\negmed\left(#1\right)}

\newcommand{\bnorm}[1]{\left\lVert#1\right\rVert}

\newcommand{\E}[1]{\ur{E}\negmed\left[#1\right]}

\newcommand{\bO}[1]{\mathcal{O}\negmed\left(#1\right)}

\def\negstrip#1 #2\relax{-#1}
\newcommand*{\negmed}{\mkern-\thinmuskip}

\theoremstyle{general} 
\theoremstyle{general} 
\theoremstyle{general} 
\theoremstyle{general} \newtheorem{proposition}{Proposition}
\theoremstyle{general} \newtheorem{p-corollary}{Corollary}[proposition]
\theoremstyle{general} 
\theoremstyle{general} 
\theoremstyle{remark}  \newtheorem{remark}{Remark}

\title{Downlink Massive MU-MIMO with Successively-Regularized Zero Forcing Precoding}
\author{\IEEEauthorblockN{Aravindh Krishnamoorthy\rlap{\textsuperscript{\IEEEauthorrefmark{2}\IEEEauthorrefmark{1}}}\,\,\,  and Robert Schober\rlap{\textsuperscript{\IEEEauthorrefmark{2}}}\\
\IEEEauthorblockA{\small \IEEEauthorrefmark{2}Friedrich-Alexander-Universit\"{a}t Erlangen-N\"{u}rnberg, Germany, 
\IEEEauthorrefmark{1}Fraunhofer Institute for Integrated Circuits (IIS) Erlangen, Germany}\vspace{-1cm}}\thanks{The authors acknowledge the financial support by the Federal Ministry of Education and Research of Germany in the programme of ``Souverän. Digital. Vernetzt.'' joint project 6G-RIC, project identification number: PIN 16KISK023.}}

\allowdisplaybreaks

\IEEEoverridecommandlockouts

\begin{document}
\maketitle

\begin{abstract}
In this letter, we consider linear precoding for downlink massive multi-user (MU) multiple-input multiple-output (MIMO) systems. We propose the novel successively-regularized zero forcing (SRZF) precoding, which exploits successive null spaces of the MIMO channels of the users, along with regularization, to control the inter-user interference (IUI) and to enhance performance and robustness to imperfect channel state information (CSI) at the base station (BS). We study the IUI characteristics of the proposed SRZF precoding for perfect and imperfect CSI at the BS. Furthermore, via computer simulations, we compare the sum rate of SRZF precoding with those of several baseline schemes including conventional and regularized zero forcing (ZF) precoding. Our simulation results reveal that, for massive MIMO systems with inter-user channel correlations, the proposed SRZF precoding significantly outperforms the considered baseline schemes for both perfect and imperfect CSI at the BS.
\end{abstract}

\section{Introduction}
Massive multi-user (MU) multiple-input multiple-output (MIMO) is a promising technology for the 6th generation (6G) and beyond communication systems where base stations (BSs) will employ a large number of antennas for transmission and reception \cite{Larsson2014}. For downlink massive MU-MIMO systems, linear zero forcing (ZF) precoding \cite{Wiesel2008,Sun2010} has been shown to be asymptotically optimal as the number of BS antennas grows large while the number of users is kept constant \cite{Larsson2014}.

However, in practical systems, ZF precoding achieves a high performance only in severely underloaded systems, where the total number of user antennas is much smaller than the number of BS antennas, and when the user channels are sufficiently distinct. On the other hand, for overloaded and critically-loaded systems, where the total number of user antennas is greater than or equal to the number of BS antennas, and when the user channels are correlated, ZF precoding suffers from poor performance, see, e.g., \cite{Krishnamoorthy2021b}. Alternatively, block diagonalization (BD) \cite{Spencer2004} has a slightly better performance than ZF precoding in critically-loaded systems. However, BD precoding suffers from poor performance in underloaded systems, see, e.g., simulation results in \cite{Krishnamoorthy2021b,Krishnamoorthy2021}.

In order to overcome the poor performance of ZF precoding and to enhance the robustness to imperfect channel state information (CSI) at the BS, Wiener filter (WF) and regularized ZF (RZF) precoding were proposed in \cite{Joham2005} and \cite{Peel2005,Bjornson2014}, respectively, where regularization constants were utilized to overcome the challenges arising from the inversion of ill-conditioned matrices. However, WF and RZF precoding introduce inter-user interference (IUI) among all users, which limits their performance, see, e.g., \cite{Krishnamoorthy2021b,Sung2009}.

On the other hand, successive null space (SNS) precoding, which exploits successive null spaces of the MIMO channel matrices of the users to enhance the performance and robustness to imperfect CSI at the BS, was proposed recently in \cite{Krishnamoorthy2021b}, in the context of MIMO rate-splitting multiple access (RSMA). SNS precoding can also be utilized as a linear precoding scheme if rate splitting and successive interference cancellation are to be avoided. In \cite{Krishnamoorthy2021b}, SNS precoding was shown to provide a balance between IUI cancellation and robustness, and achieved superior performance compared to BD, ZF, and RZF precoding for both perfect and imperfect CSI at the BS. Unfortunately, the algorithm for computing the SNS precoders provided in \cite{Krishnamoorthy2021b} entails a comparatively high computational complexity, limiting its practical applicability.

In this letter, we develop a novel, low-complexity successively-regularized zero forcing (SRZF) linear precoding scheme, which exploits successive null spaces of the MIMO channels of the users, along with regularization, to control the IUI and enhance performance and robustness to imperfect CSI at the BS. The main contributions of this letter are as follows.
\begin{itemize}
	\item We introduce the proposed novel SRZF precoding scheme.
	\item We study the IUI characteristics of SRZF precoding for perfect and imperfect CSI at the BS.
	\item We compare the sum rate (SR) performance of the proposed SRZF precoding with those of SNS \cite{Krishnamoorthy2021b}, BD \cite{Spencer2004}, ZF \cite{Wiesel2008}, RZF \cite{Peel2005,Bjornson2014}, and WF \cite{Joham2005} precoding via simulations.
\end{itemize}
Our simulation results reveal that, for massive MIMO systems with inter-user MIMO channel correlations and for both perfect and imperfect CSI at the BS, the proposed SRZF precoding scheme significantly outperforms conventional BD, ZF, RZF, and WF precoding, and provides a low-complexity alternative to SNS precoding with a moderate loss in performance.

The remainder of this letter is organized as follows. In Section \ref{sec:prelim}, we provide the system model and a brief review of SNS precoding \cite{Krishnamoorthy2021b}. In Section \ref{sec:proposed}, we describe the proposed SRZF precoding and decoding schemes. The IUI characteristics of SRZF precoding are analyzed in Section \ref{sec:pa}. Simulation results are presented in Section \ref{sec:sim}, and the letter is concluded in Section \ref{sec:con}.

\emph{Notation:} Boldface capital letters $\boldsymbol{X}$ and boldface lowercase letters $\boldsymbol{x}$ denote matrices and vectors, respectively. $\boldsymbol{X}^\ur{H}$, $\boldsymbol{X}^{+}$, and tr($\boldsymbol{X}$) denote the Hermitian transpose, pseudo-inverse, and trace of matrix $\boldsymbol{X}$, respectively. $\mathbb{C}^{m \times n}$ and $\mathbb{R}^{m \times n}$ denote the sets of $m \times n$ matrices with complex-valued and real-valued entries, respectively. The $(i,j)$-th entry of matrix $\boldsymbol{X}$ is denoted by $[\boldsymbol{X}]_{i,j}$ and the $i$-th element of vector $\boldsymbol{x}$ is denoted by $[\boldsymbol{x}]_{i}$. $\boldsymbol{I}_{n}$ denotes the $n \times n$ identity matrix. The circularly symmetric complex-valued Gaussian distribution with mean $\boldsymbol{\mu}$ and covariance matrix $\boldsymbol{\Sigma}$ is denoted by $\mathcal{CN}(\boldsymbol{\mu},\boldsymbol{\Sigma})$; $\sim$ stands for ``distributed as". $\E{\cdot}$ denotes statistical expectation.

\section{Preliminaries}
\label{sec:prelim}
In this section, we present the considered downlink MIMO system model, the imperfect CSI model, and a brief review of SNS precoding \cite{Krishnamoorthy2021b}.

\subsection{System Model}
\label{sec:sys_model}
We consider an \emph{underloaded} or \emph{critically loaded} downlink MU-MIMO communication system comprising a BS with $N$ transmit antennas and $K$ users equipped with $M_k,k=1,\dots,K,$ antennas such that $N \geq \sum_{k=1}^{K} M_k = M,$ where $M$ denotes the total number of user antennas.

We assume that $M_k$ streams are transmitted to user $k, k=1,\dots,K.$ If fewer streams are desired, then, user $k$'s MIMO channel can be transformed into a matrix with fewer effective receive antennas via singular value decomposition (SVD) \cite{Sun2010} and $M_k$ can be redefined accordingly. Let $\bo{s}_{k} \in \mathbb{C}^{M_k\times 1}$ denote the MIMO symbol vector of user $k$ satisfying $\E{\bo{s}_{k} \bo{s}_{k}^\H} = \bo{I}_{M_k}, k=1,\dots,K.$ We assume that the $[\bo{s}_{k}]_i, i=1,\dots,M_k, k=1,\dots,K,$ are independent. For transmission, $\bo{s}_{k}$ is precoded using linear precoder $\bo{P}_{k} \in \mathbb{C}^{N\times M_k}, k=1,\dots,K,$ to obtain overall transmit signal $\bo{x} = \sum_{k = 1}^{K} \bo{P}_{k}\bo{s}_{k}.$ The power constraint at the BS is given by
\begin{align}
	\sum_{k = 1}^{K} \btr{\bo{P}_{k} \bo{P}_{k}^\H} \leq P_\T, \label{eqn:pt}
\end{align}
where $P_\ur{T}$ denotes the available transmit power. Let $\frac{1}{\sqrt{\mathstrut L_k}} \bo{H}_k \in \mathbb{C}^{M_k\times N}$ denote the MIMO channel matrix between the BS and user $k,$ where scalar $L_k$ models the path loss between the BS and user $k,$ and $\bo{H}_k$ models the small scale fading. Here, we assume that all MIMO channel matrices have full row rank\footnote{We note that a row-rank deficient MIMO channel matrix can be transformed into a full row-rank matrix with fewer effective receive antennas via SVD, see, e.g., \cite{Sun2010}, \cite[App. C]{Scutari2009}.}. Then, the received signal at user $k$ is given by
\begin{align}
	\bo{y}_k = \frac{1}{\sqrt{\mathstrut L_k}} \bo{H}_k\bo{x} + \bo{z}_k = \frac{1}{\sqrt{\mathstrut L_k}} \bo{H}_k \sum_{k' = 1}^{K} \bo{P}_{k'}\bo{s}_{k'} + \bo{z}_k, \label{eqn:y_k}
\end{align}
where $\bo{z}_k \in \mathbb{C}^{M_k\times 1} \sim \mathcal{CN}(\bo{0}, \sigma^2\bo{I}_{M_k})$ denotes the complex additive white Gaussian noise (AWGN) vector at user $k.$

\subsection{Imperfect CSI Model}
\label{sec:impmimo}
\Copy{Imperfect}{In this letter, we assume that only quantized and outdated MIMO channel matrices $\bar{\bo{H}}_k,k=1,\dots,K,$ are available at the BS for computing the precoders. Nevertheless, we assume that the BS knows\footnote{This assumption is motivated by the slow variation of path loss $L_k$ and the resulting low feedback requirement.} scalar $L_k$ perfectly, and the users know their own MIMO channel matrices perfectly\footnote{In practice, the MIMO channel matrices can be estimated frequently and accurately at the receivers by exploiting, e.g., high-power pilot sequences transmitted by the BS.}. We model the imperfect MIMO channel matrices at the BS as $\frac{1}{\sqrt{\mathstrut L_k}} \bar{\bo{H}}_k = \frac{1}{\sqrt{\mathstrut L_k}} \bo{H}_k + \frac{1}{\sqrt{\mathstrut L_k}} \Delta\bo{H}_k,$ for $k=1,\dots,K,$ where $\bo{H}_k$ denotes the actual MIMO channel matrix of user $k,$ and $\Delta\bo{H}_k \in \mathbb{C}^{M_k\times N}$ models the CSI error at the BS \cite{Wang2007}.}

\subsection{SNS Precoding}
\label{sec:snsprec}
Let $\bo{\Psi}_k \in \mathbb{C}^{N\times N_k},$ $N_k = N-\sum_{k' = 1}^{k-1} M_{k'},$ denote a matrix whose columns are the unit-length basis vectors of the null space of the following augmented matrix: $\bo{F}_k = \begin{bmatrix} \bo{H}_{1}^\T & \bo{H}_{2}^\T & \dots & \bo{H}_{k-1}^\T\end{bmatrix}^\T \in \mathbb{C}^{(N-N_k)\times N},$ with the convention $\bo{\Psi}_{1} = \bo{I}_N.$ The SNS precoder \cite{Krishnamoorthy2021b} for user $k$ is constructed as $\bo{P}_{k} = \bo{\Psi}_{k} \bo{W}_{k}^\frac{1}{2},$ where $\bo{W}_{k} \in \mathbb{C}^{N_k\times M_k}$ is a rectangular matrix, which is used for combining the column vectors of $\bo{\Psi}_k$ into the precoding vectors and for power allocation (PA). Furthermore, for SNS precoding, the power constraint in (\ref{eqn:pt}) can be rewritten as $\sum_{k = 1}^{K} \btr{\bo{W}_{k}\bo{W}_{k}^\H} \leq P_\T.$ Now, with SNS precoding, the received signal in (\ref{eqn:y_k}) simplifies to $\bo{y}_k = \frac{1}{\sqrt{\mathstrut L_k}} \bo{H}_k \sum_{k' = 1}^{k} \bo{P}_{k'}\bo{s}_{k'} + \bo{z}_k.$ That is, user $k$ experiences IUI from users $1,\dots,k-1,$ and causes IUI to users $k+1,\dots,K,$ which, for the optimal $\bo{W}_{k}, k=1,\dots,K,$ enhances performance and robustness compared to BD, ZF, and RZF precoding, see \cite{Krishnamoorthy2021b} for details. However, obtaining the optimal $\bo{W}_{k}, k=1,\dots,K,$ via \cite[Algorithm 1]{Krishnamoorthy2021b} is computationally complex, limiting the applicability of SNS precoding in practice.

\section{Proposed SRZF Precoding and Decoding Schemes}
\label{sec:proposed}
In this section, we introduce the proposed SRZF precoding and decoding schemes and determine the resulting achievable user rates.

\subsection{Precoding Scheme}
\label{sec:prec}
\Copy{Precoder}{Let $\bar{\bo{H}} = \begin{bmatrix} \bar{\bo{H}}_1^\T & \dots & \bar{\bo{H}}_K^\T \end{bmatrix}^\T \in \mathbb{C}^{M\times N}.$ Furthermore, let $m_k = \sum_{k'=1}^{k-1} M_{k'} + 1,$  with $m_1 = 1,$ denote the index of the first row of user $k$'s channel in $\bar{\bo{H}},$ defined above. Moreover, let $\bo{J}_k \in \mathbb{R}^{M\times M}$ denote a real-valued, diagonal matrix such that $[\bo{J}_k]_{i_k,i_k} = 1,$ $i_k = m_{k},\dots,M,$ $k=1,\dots,K-1,$ and the remaining elements are zeros, with the convention $\bo{J}_K = \bo{0}.$ Then, the proposed precoder matrix for user $k,$ $\bo{P}_k \in \mathbb{C}^{N\times M_k},$ is chosen as:
\vspace{-0.35cm}
\begin{align}
	\bo{P}_k = \bar{\bo{H}}^\H \bar{\bo{\Phi}}_k \bo{D}_k^\frac{1}{2}, \label{eqn:pk}
\end{align}
where $\bar{\bo{\Phi}}_k = \begin{bmatrix}\bar{\bo{\phi}}_{k,1},\dots, \bar{\bo{\phi}}_{k,M_k}\end{bmatrix} \in \mathbb{C}^{N\times M_k}$ contains columns $m_k,\dots,m_k+M_k-1$ of matrix $(\bar{\bo{H}}\bar{\bo{H}}^\H + \alpha_k \bo{J}_k)^{-1},$ $\alpha_k > 0,k=1,\dots,K,$ is a small regularization constant whose role is described later, and $\bo{D}_k \in \mathbb{R}^{M_k \times M_k} \succcurlyeq \bo{0}, k=1,\dots,K,$ is the diagonal power-allocation matrix, which based on (\ref{eqn:pt}) and (\ref{eqn:pk}) has to satisfy the power constraint:
\begin{align}
	\sum_{k = 1}^{K} \btr{\bar{\bo{H}}^\H \bar{\bo{\Phi}}_k \bo{D}_k \bar{\bo{\Phi}}_k^\H \bar{\bo{H}}} \leq P_\T. \label{eqn:pt1}
\end{align}}

\begin{remark}
	SRZF precoding reduces to ZF precoding for $\alpha_k = 0, k=1,\dots,K,$ and to RZF precoding if $\alpha_k \bo{J}_k$ is replaced by $\alpha_k \bo{I}_M, k=1,\dots,K.$
\end{remark}

\Copy{Future1}{
\begin{remark}
	\label{rem:gsrzf}
	A more general version of SRZF precoding can be obtained by utilizing an arbitrary (diagonal) regularization matrix instead of $\alpha_k \bo{J}_k,$ which enables a greater control of the IUI, see Section \ref{sec:pa} for details; see also the banded structure proposed in \cite{Hu2017}.
\end{remark}
}

\subsection{Decoding at the Receivers}
Substituting (\ref{eqn:pk}) into (\ref{eqn:y_k}), the received signal at user $k$ can be simplified to: $\bo{y}_k = \frac{1}{\sqrt{\mathstrut L_k}} \bar{\bo{G}}_{k,k} \bo{D}_{k}^\frac{1}{2} \bo{s}_{k} + \bo{u}_{k} + \bo{z}_k,$ where $\bar{\bo{G}}_{k,k'} = \bo{H}_k \bar{\bo{H}}^\H \bar{\bo{\Phi}}_{k'} \in \mathbb{C}^{M_k\times M_k}, k=1,\dots,K,$ is the effective MIMO channel of user $k$ for the symbol vector of user $k'$ and $\bo{u}_{k} = \frac{1}{\sqrt{\mathstrut L_k}} \sum_{\substack{k' = 1\\k' \neq k}}^{K}  \bar{\bo{G}}_{k,k'} \bo{D}_{k'}^\frac{1}{2} \bo{s}_{k'} \in \mathbb{C}^{M_k \times 1}$ is the IUI for user $k.$ Here, at user $k,$ based on $\bo{y}_k$ above, the elements of $\bo{s}_k$ are decoded \emph{jointly} treating $\bo{u}_{k}$ as noise.

\subsection{Achievable Rates and Sum Rate}
\label{sec:ar}
Based on $\bo{y}_k,$ the achievable rate for the symbol vector of user $k,$ $k=1,\dots,K,$ is given as follows: $R_k = \log_2\det\Big(\bo{I}_{M_k} + \frac{1}{L_k} \bar{\bo{G}}_{k,k} \bo{D}_{k} \bar{\bo{G}}_{k,k}^\H \bar{\bo{N}}_k^{-1}\Big),$ where $\bar{\bo{N}}_k = \sigma^2 \bo{I}_{M_k} + \frac{1}{L_k} \sum_{\substack{k' = 1, k' \neq k}}^{K} \bar{\bo{G}}_{k,k'} \bo{D}_{k'} \bar{\bo{G}}_{k,k'}^\H.$ Thus, the SR is given by $R_\mathrm{sr} = \sum_{k = 1}^{K} R_k.$

\subsection{Fixed Power Allocation}
\label{sec:fpa}
\Copy{FPA}{In this letter, for simplicity, we restrict ourselves to fixed power allocation (FPA). For FPA, based on \cite{Bjornson2014}, we choose
\begin{align}
	[\bo{D}_k]_{l,l} = {P_\T}/{\big(M\operatorname{tr}\negmed\big(\bar{\bo{H}}^\H \bar{\bo{\phi}}_{k,l} \bar{\bo{\phi}}_{k,l}^\H \bar{\bo{H}}\big)\big)}, \label{eqn:fpa}
\end{align}
for $k=1,\dots,K,$ $l=1,\dots,M_k,$ where $\bar{\bo{\phi}}_{k,l}$ is the $l$-th column of $\bar{\bo{\Phi}}_k,$ as defined in Section \ref{sec:prec}. In (\ref{eqn:fpa}), the normalization by $\btr{\bar{\bo{H}}^\H \bar{\bo{\phi}}_{k,l} \bar{\bo{\phi}}_{k,l}^\H \bar{\bo{H}}}$ ensures that all precoding vectors of all users have the same norm. Furthermore, the normalized precoding vectors are allocated a power of $\frac{P_\T}{M}.$ It is easy to verify that the proposed FPA in (\ref{eqn:fpa}) satisfies the power constraint in (\ref{eqn:pt1}).

\begin{remark}
	We note that the performance of the proposed SRZF precoding can be further improved by optimizing the power allocation, see, e.g., \cite[Sec. V-D]{Krishnamoorthy2021b}.
\end{remark}
}

\subsection{Computational Complexity}
For the proposed SRZF scheme with FPA, computing $\bar{\bo{\Phi}}_K$ in (\ref{eqn:pk}) via Cholesky decomposition and equation solving entails a total complexity of $\bO{MN^2 + \frac{1}{6}M^3 + \frac{1}{3}M_K M^2}$ \cite{Krishnamoorthy2013}.
Next, computing $\bar{\bo{\Phi}}_k,k=1,\dots,K-1,$ entails a complexity of $\mathcal{O}\big(MN^2 + \frac{1}{2}M^2 + \frac{1}{3}(M$ ${}-M_K)M^2\big).$ Hence, SRZF precoding entails an overall complexity of $\mathcal{O}\Big(2 MN^2 + \frac{1}{2}M^3$ ${}+ \frac{1}{2}M^2\Big),$ which is only marginally higher than the complexity for ZF, RZF, and WF precoding, given by $\bO{2 MN^2 + \frac{1}{2} M^3},$ and BD precoding, given by $\mathcal{O}\Big(2 MN^2 +  \sum_{k=1}^{K}M_k^3\Big),$ see \cite{Krishnamoorthy2021b} for details.

On the other hand, computing the precoding vectors and PA for SNS precoding in Section \ref{sec:snsprec} via \cite[Algorithm 1]{Krishnamoorthy2021b} entails a significantly higher computational complexity of $\bO{2 MN^2 + 120 \big(\sum_{k=1}^{K} N_k^2\big)^\frac{3}{2}}.$

\section{IUI Analysis of SRZF Precoding}
\label{sec:pa}
In this section, we study the IUI characteristics of SRZF precoding for perfect and imperfect CSI.

\subsection{Interference Analysis for Perfect CSI}
\label{sec:iap}

Let $\breve{\bo{P}}_{k} = \bo{P}_{k} \mid_{\bar{\bo{H}} \to \bo{H}}$ and $\breve{\bo{\Phi}}_{k'} = \bar{\bo{\Phi}}_{k} \mid_{\bar{\bo{H}} \to \bo{H}}$ denote the matrices corresponding to $\bo{P}_k$ and $\bar{\bo{\Phi}}_k,$ respectively, for perfect CSI, i.e., $\Delta\bo{H}_k = \bo{0}, k=1,\dots,K.$ In the following proposition, we show that the precoder in (\ref{eqn:pk}) is a special case of SNS precoding \cite{Krishnamoorthy2021b}.
\begin{proposition}
	\label{prop:sns}
	We have
	\begin{align}
		\bo{H}_k\breve{\bo{P}}_{k'} = \begin{cases}
			-\alpha_{k'}\breve{\bo{\Omega}}_{k'}\bo{D}_{k'}^\frac{1}{2} & \text{if $k' < k$} \\
			(\bo{I}_{M_{k}} -\alpha_{k}\breve{\bo{\Omega}}_{k})\bo{D}_k^\frac{1}{2} & \text{if $k' = k$} \\
			\bo{0} & \text{if $k' > k,$}
		\end{cases} \label{eqn:sns:hkpksns}
	\end{align}
	where $\breve{\bo{\Omega}}_{k'} \in \mathbb{C}^{M_k \times M_{k'}}$ contains rows $m_k,\dots,m_k+M_k-1$ of $\breve{\bo{\Phi}}_{k'}.$	
\end{proposition}
\begin{proof}
	Please see Appendix \ref{app:sns}.
\end{proof}

Based on (\ref{eqn:y_k}) and Proposition \ref{prop:sns}, the received signal at user $k$ can be simplified to
\begin{align}
	\scalemath{0.95}{\bo{y}_k = (\bo{I}_{M_{k}} -\alpha_{k}\breve{\bo{\Omega}}_{k})\bo{D}_k^\frac{1}{2}\bo{s}_k - \sum_{k'=1}^{k-1} \alpha_{k'}\breve{\bo{\Omega}}_{k'}\bo{D}_{k'}^\frac{1}{2}\bo{s}_{k'} + \bo{z}_k.} \label{eqn:y_k_srzf}
\end{align}
\Copy{GenRZF}{As seen from (\ref{eqn:y_k_srzf}), SRZF precoding is a special case of SNS precoding \cite{Krishnamoorthy2021b}, described earlier in Section \ref{sec:snsprec}. That is, for SRZF precoding, user $k$ experiences IUI only from users $1,\dots,k-1,$ and causes IUI to users $k+1,\dots,K.$ Furthermore, the IUI for user $k$ caused by users $1,\dots,k-1$ can be adjusted using the regularization constants $\alpha_k,k=1,\dots,K.$ Thus, SRZF precoding enables flexible interference management by combining the IUI cancellation capability offered by ZF precoding with the robustness of RZF precoding to inversion of ill-conditioned matrices.}

\Copy{WhyReg}{In this letter, for simplicity, analogous to the regularization constants proposed for RZF precoding \cite{Peel2005,Bjornson2014}, we set $\alpha_k = \frac{M\sigma^2}{P_\T},$ $k=1,\dots,K.$ Nevertheless, we note that the performance of the proposed SRZF precoders may be further enhanced by optimizing the values of $\alpha_k,k=1,\dots,K,$ for a given deployment scenario.}

\begin{remark}
	As seen from Proposition \ref{prop:sns}, the IUI experienced by a user depends on the user index. Hence, similar to SNS precoding, SRZF precoding is also sensitive to user ordering \cite{Krishnamoorthy2021b}. In this letter, we utilize the suboptimal fixed user-index permutation described in \cite{Krishnamoorthy2021b},  where the users are permuted in descending order of their single-user achievable rates. Nevertheless, the performance of SRZF precoding can be further improved by determining the optimal user ordering, e.g., via an exhaustive search.
\end{remark}

\subsection{Interference Analysis for Imperfect CSI}
\label{sec:pai}
\Copy{Robustness}{Next, in the following proposition, we study the robustness of SRZF precoding to imperfect CSI.

\begin{proposition}
	\label{prop:anaimp}
	Let $\breve{\bo{P}}_{k'}^{(k)} = \bo{P}_{k'} \mid_{\bar{\bo{H}}_k \to \bo{H}_k},$ $k'= 1,\dots,K,$ denote the SRZF precoders when the CSI of user $k$ at the BS is perfect and that of users $k', k' \neq k,$ $k'=1,\dots,K,$ is imperfect. Then, we have for $k' > k,$ 
	\begin{align}
		\bnorm{\bo{H}_k (\bo{P}_{k'} - \breve{\bo{P}}_{k'}^{(k)})} \leq \bnorm{\Delta \bo{H}_k} \bnorm{\bo{P}_{k'}}.
	\end{align}
\end{proposition}

\begin{proof}
	Please see Appendix \ref{app:anaimp}.
\end{proof}

From Proposition \ref{prop:anaimp}, we note that small imperfections in the CSI of user $k$ result in small additional interference from users $k+1,\dots,K,$ i.e., up to $\bO{\bnorm{\Delta \bo{H}_k}}.$ Furthermore, it can be shown that the additional interference experienced by users $k+1,\dots,K$ due to user $k$'s imperfect CSI is also small, see \cite[Prop. 3]{Krishnamoorthy2021b}. On the other hand, the performance of ZF precoding is significantly degraded by imperfect CSI. Furthermore, for RZF precoding, poor CSI quality of a user causes additional IUI to all users irrespective of their CSI quality, see Figure \ref{fig:ik} in Section \ref{sec:sim}.}

\section{Simulation Results}
\label{sec:sim}
In this section, we first validate the analytical expression in Proposition \ref{prop:anaimp}. Next, we compare the SR of the proposed SRZF precoding with the SRs of SNS \cite{Krishnamoorthy2021b}, conventional BD \cite{Spencer2004}, ZF \cite{Wiesel2008}, RZF \cite{Peel2005,Bjornson2014}, and WF \cite{Joham2005} precoding for perfect and imperfect CSI via computer simulations. 

\Copy{SystemSetup}{For our simulations, we utilize the channel model proposed in \cite{Raghavan2017} to capture the correlations in the MIMO channels of the users based on their angular positions. We consider a \emph{critically-loaded} system with $K=64$ users, $M_k=2$ antennas per user, and $N=128$ BS antennas\footnote{The results presented in this paper are also valid for underloaded systems. Furthermore, we note that for systems without significant inter-user MIMO channel correlations, which can not benefit from the enhanced robustness to inversion of ill-conditioned matrices, the performance of the proposed SRZF precoding is similar to that of RZF precoding.}. The users $1,\dots,\frac{K}{2}$ and $\frac{K}{2}+1,\dots,K$ spread uniformly around the BS and are located at distances of $d_k = 50\text{ m,}$ $k = 1,\dots,\frac{K}{2}$ and $d_k = 250\text{ m,}$ $k = \frac{K}{2}+1,\dots,K,$ respectively. Furthermore, users $\frac{K}{2}+1,\dots,K,$ are located within an angular spread of $0.5^\circ$ from users $1,\dots,\frac{K}{2},$ respectively. This causes a high correlation of the MIMO channels of users $k$ and $\frac{K}{2}+k,$ $k=1,\dots,\frac{K}{2}.$ The scalar path loss for user $k, L_k,$ is set to $d_k^2.$ Furthermore, for imperfect CSI, the additive i.i.d. Gaussian error model, where the elements of $\bo{\Delta H}_k,k=1,\dots,K,$ have zero mean and variance $\mu_k^2,$ is adopted. Furthermore, we assume that $\Delta\bo{H}_k$ and $\bo{H}_k$ are statistically independent. We assume $\sigma^2 = -35 \text{ dBm,}$ and utilize FPA for the proposed and the baseline schemes, except for SNS precoding, for which we adopt the baseline successive convex approximation (SCA)-based scheme given in \cite{Krishnamoorthy2021b}.}

\begin{figure}
	\centering
	\includegraphics[width=0.75\columnwidth]{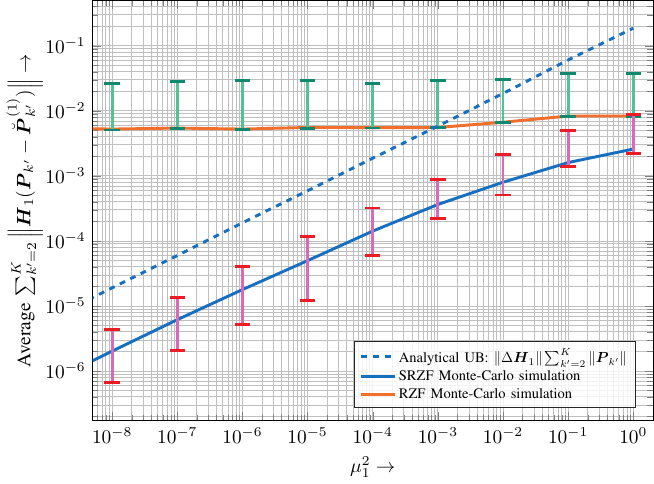}
	\caption{Average $\sum_{k'=2}^{K}\bnorm{\bo{H}_1(\bo{P}_{k'} - \breve{\bo{P}}_{k'}^{(1)})}$ with $K=64,$ $N=128,$ $M_k=2,\forall\,k,$ as a function of $\mu_1^2$ with $d_k = 50\text{ m,}\,\forall\,k,$ and $\mu_k^2 = 10^{-2},\forall\,k\neq 1.$ Error bars show the range of values encountered.}
	\label{fig:ik}
\end{figure}

\Copy{PropValidation}{
In Figure \ref{fig:ik}, the additional IUI experienced by user $1$ due to imperfect CSI at the BS, i.e., $\sum_{k'=2}^{K}\bnorm{\bo{H}_1(\bo{P}_{k'} - \breve{\bo{P}}_{k'}^{(1)})},$ obtained via Monte-Carlo simulation of $10^4$ channel realizations is shown as a function of $\mu_1,$ along with the analytical upper bound (UB) provided in Proposition \ref{prop:anaimp}. Here, we have assumed $P_\T = 30 \text{ dBm},$ $\mu^2_k = 10^{-2},$ $k=2,\dots,K.$ From the figure, we note that the analytical UB accurately bounds the total additional IUI. Furthermore, the average IUI is close to the minimum indicating that, in most cases, the additional IUI is small. On the other hand, for RZF, the CSI quality of user $1$ has negligible impact on its IUI, which is dominated by the IUI from users $k' = 2,\dots,K.$
}

\begin{figure}
	\centering
	\includegraphics[width=0.75\columnwidth]{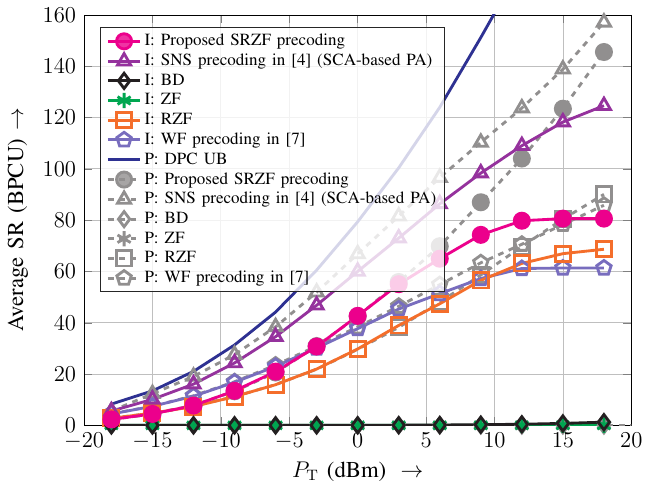}
	\caption{Average SR as a function of $P_\T$ for a critically-loaded system with $K=64,$ $N=128,$ $M_k=2,$  $d_k = 50\text{ m,}$ for $k=1,\dots,\frac{K}{2},$ $d_k = 250\text{ m,}$ for $k=\frac{K}{2}+1,\dots,K,$ $\mu_k^2 = 10^{-2},$ for $k = 1,\dots,K,$ $\sigma^2 = -35 \text{ dBm,}$ the MIMO channel model in \cite{Raghavan2017}, and FPA\protect\footnotemark. P: Perfect CSI, I: Imperfect CSI.}
	\label{fig:64}
	\vspace{-0.25cm}
\end{figure}
\footnotetext{The curves for SNS precoding in \cite{Krishnamoorthy2021b} utilize SCA for determining the precoding vectors and for PA.}

\Copy{RevisedResults}{Next, in Figure \ref{fig:64}, for obtaining the SR curves, we assume $\mu_k = 10^{-2},$ $k=1,\dots,K.$ From the figure, we observe that the proposed SRZF precoding scheme significantly outperforms RZF and WF precoding for both perfect and imperfect CSI, as it can benefit from both IUI cancellation and robustness to inversion of ill-conditioned matrices, thereby providing an attractive alternative to these conventional schemes for deployment in massive MU-MIMO systems. On the other hand, BD and ZF precoding have a very poor performance with an SR close to zero due to the high correlation in the MIMO channels of the users. WF precoding has a marginally higher performance than SRZF precoding at very low SNRs, as in this case, the WF precoder reduces to a matched filter \cite{Joham2005}. Lastly, for perfect CSI, the SNS precoding from \cite{Krishnamoorthy2021b}, which utilizes high-complexity SCA-based precoding vector selection and PA \cite{Krishnamoorthy2021b}, achieves a moderate SR increase over the proposed SRZF precoding scheme, which entails a substantially lower computational complexity. However, for imperfect CSI, the performance of SNS precoding from \cite{Krishnamoorthy2021b} is higher, owing to the optimized IUI management, compared to SRZF precoding, which utilizes fixed PA and fixed regularization constants.}

\section{Conclusion}
\label{sec:con}
In this letter, we considered linear precoding for downlink massive MU-MIMO systems. We proposed the novel SRZF precoding scheme, which exploits successive null spaces of the MIMO channels of the users, along with regularization, to enhance performance and robustness. We showed in Proposition \ref{prop:sns} that SRZF precoding is a special case of SNS precoding \cite{Krishnamoorthy2021b} and combines the IUI cancellation offered by ZF precoding with the robustness to inversion of ill-conditioned matrices of RZF precoding. Furthermore, we provided an UB for the additional IUI experienced by a user if the CSI is imperfect  at the BS. We compared the performance of the proposed scheme with those of SNS, BD, ZF, RZF, and WF precoding via computer simulations. Our results showed that for a critically-loaded system with a high correlation in the MIMO channels experienced by users and perfect and imperfect CSI at the BS, the proposed SRZF precoding outperforms the conventional schemes and provides a low-complexity alternative to SNS precoding with moderate loss in performance, making it an attractive option for downlink massive MU-MIMO systems.

\Copy{Future2}{The optimization of the regularization factors and user index permutation, taking into account the deployment scenario and CSI quality of the users, is an interesting topic for future study for further enhancing the performance and robustness of SRZF precoding. Furthermore, as mentioned in Remark \ref{rem:gsrzf}, a generalization of SRZF precoding to arbitrary diagonal regularization matrices to enhance performance and robustness is also of interest.}
\begin{appendices}
\renewcommand{\thesection}{\Alph{section}}
\renewcommand{\thesubsection}{\thesection.\arabic{subsection}}
\renewcommand{\thesectiondis}[2]{\Alph{section}:}
\renewcommand{\thesubsectiondis}{\thesection.\arabic{subsection}:}
\section{Proof of Proposition \ref{prop:sns}}
\label{app:sns}
\begin{proof}[\unskip\nopunct]
Since $(\bo{H}\bo{H}^\H + \alpha_{k'} \bo{J}_{k'}) (\bo{H}\bo{H}^\H + \alpha_{k'} \bo{J}_{k'})^{-1} = \bo{I}_M,$
$k' = 1,\dots,K,$ we have
\begin{align}
	&\bo{H}\bo{H}^\H(\bo{H}\bo{H}^\H + \alpha_{k'} \bo{J}_{k'})^{-1} \nonumber\\
	&\qquad {}= \bo{I}_M - \alpha_{k'} \bo{J}_{k'}(\bo{H}\bo{H}^\H + \alpha_{k'} \bo{J}_{k'})^{-1}. \label{eqn:app:sns:i}
\end{align}
Now, selecting the columns $m_{k'},\dots,m_{k'}+M_{k'}-1$ of the matrices on the left and right hand side of (\ref{eqn:app:sns:i}), we obtain:
\begin{align}
	\bo{H}\bo{H}^\H \breve{\bo{\Phi}}_{k'} = \begin{bmatrix}
		\bo{0}_{\sum_{k''=1}^{k'-1} M_{k''}\times M_{k'}} \\
		\bo{I}_{M_{k'}}  \\
		\bo{0}  
	\end{bmatrix} - \alpha_{k'} \bo{J}_{k'} \breve{\bo{\Phi}}_{k'}. \label{eqn:app:sns:hpk}
\end{align}
Next, selecting rows $m_k,\dots,m_k+M_k-1$ of the matrices on the left and right hand side of (\ref{eqn:app:sns:hpk}), noting that the first $\sum_{k''=1}^{k'-1} M_{k''}$ rows of $\bo{J}_{k'} \breve{\bo{\Phi}}_{k'}$ are zeros, and utilizing the definitions of $\bo{P}_k$ and $\breve{\bo{P}}_k$ in (\ref{eqn:pk}) and Section \ref{sec:iap}, respectively, we obtain (\ref{eqn:sns:hkpksns}). 
\end{proof}

\section{Proof of Proposition \ref{prop:anaimp}}
\label{app:anaimp}
\begin{proof}[\unskip\nopunct]
We note that, for $k' > k,$ $k,k'= 1,\dots,K,$
\begin{align}
	\scalemath{0.95}{\bnorm{\bo{H}_k(\bo{P}_{k'} - \breve{\bo{P}}_{k'}^{(k)})}} &{}\scalemath{0.95}{\overset{(a)}{=} \bnorm{\bo{H}_k\bo{P}_{k'}} = \bnorm{(\bar{\bo{H}}_k - \Delta\bo{H}_k)\bo{P}_{k'}}} \nonumber\\ &{}\scalemath{0.95}{\overset{(b)}{\leq} \bnorm{\Delta \bo{H}_k} \bnorm{\bo{P}_{k'}}},
\end{align}
where (a) and (b) hold because of Proposition \ref{prop:sns} with the following substitutions: For (a), $\bo{H}_{k'} \to \bar{\bo{H}}_{k'},$ $k' \neq k,$ and $\breve{\bo{P}}_{k'} \to \breve{\bo{P}}_{k'}^{(k)},$ $\forall\,k',$ leading to $\bo{H}_k \breve{\bo{P}}_{k'}^{(k)} = \bo{0}$ for $k' > k,$ and, for (b), $\bo{H}_{k'} \to \bar{\bo{H}}_{k'}$ and $\breve{\bo{P}}_{k'} \to \bo{P}_{k'},$ $\forall\,k',$ leading to $\bar{\bo{H}}_k\bo{P}_{k'} = \bo{0}$ for $k' > k.$
\end{proof}


\end{appendices}

\bibliographystyle{IEEEtran}
\bibliography{IEEEabrv,references}

\includepdf[pages=-]{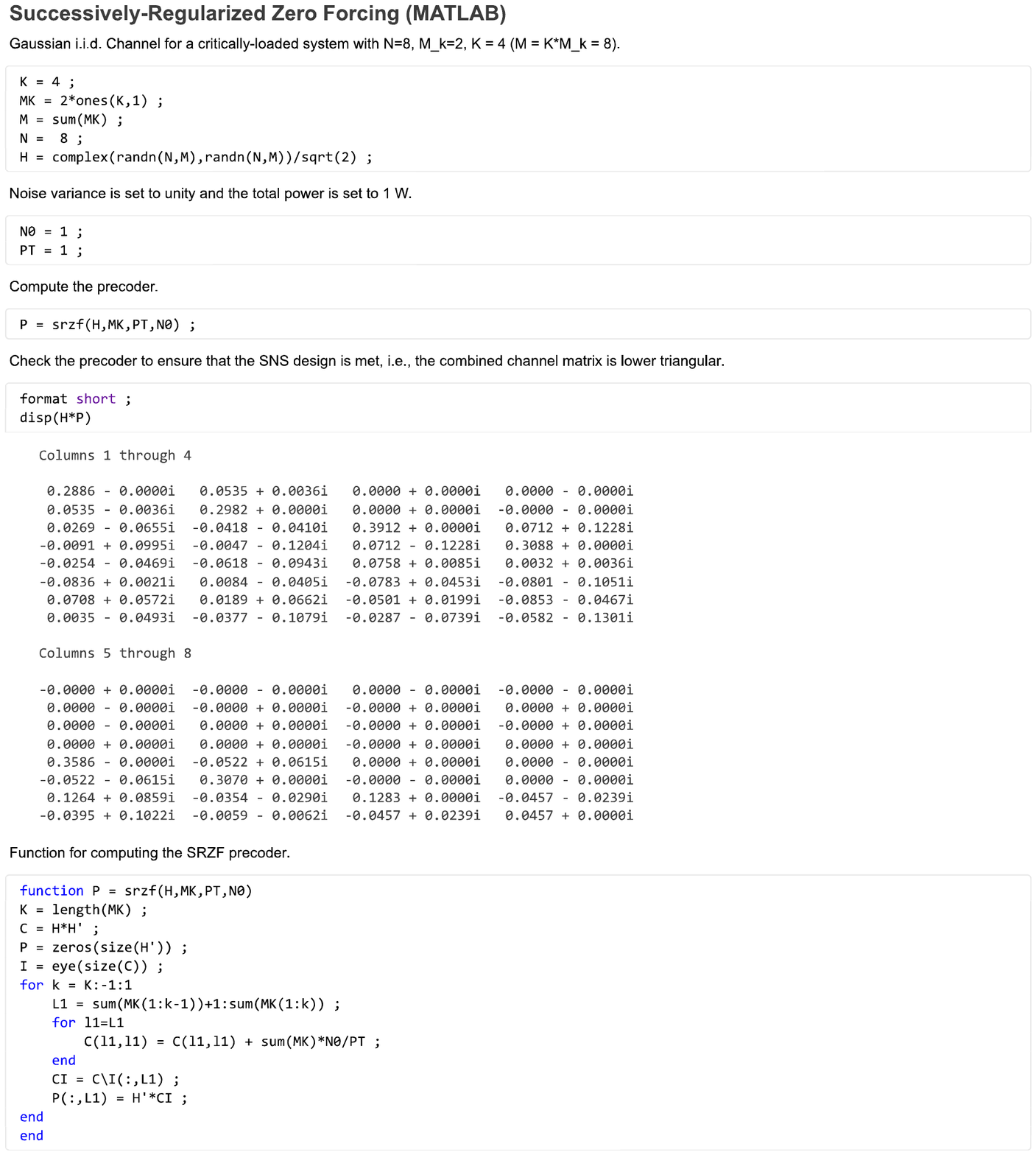}

\end{document}